\pdfoutput=1

\documentclass[3p,sort&compress]{elsarticle}

\usepackage[usenames,dvipsnames,svgnames,table]{xcolor}
\usepackage{subfig}
\DeclareGraphicsExtensions{.pdf}
\usepackage{amsmath,mathtools,amsthm,amssymb}
\usepackage{placeins}
\usepackage{algorithm}
\usepackage[noend]{algpseudocode}
\usepackage[inline]{enumitem}
\usepackage{siunitx}
\DeclarePairedDelimiter{\parens}{(}{)}

\DeclareMathOperator*{\phaseop}{phase}
\newcommand{\phase}[1]{\phaseop\parens{#1}}
\DeclareMathOperator{\real}{Re}
\DeclareMathOperator{\imag}{Im}
\theoremstyle{plain}

\newtheorem{myprp}{Proposition}
\newtheorem{mylma}{Lemma}
\theoremstyle{remark}

\newcommand{\prn}[1]{\left ( #1 \right )}
\newcommand{\brq}[1]{[ #1 ]}
\newcommand{\prns}[1]{( #1 )}
\newcommand{\sym}[1]{#1^{\prn{\text{s}}}}
\newcommand{\set}[1]{\left \{ #1 \right \}}

\newcommand{\dg}{^{\prn{g}}}
\newcommand{\nr}{_\text{N}}
\newcommand{\dz}{^{\prn{0}}}
\newcommand{\im}{\mathrm{i}}

\newcommand{\R}{\mathbb{R}}
\newcommand{\C}{\mathbb{C}}

\newcommand{\Z}{\mathbb{Z}}

\DeclareMathOperator*{\opU}{U}
\DeclareMathOperator*{\opSO}{SO}
\newcommand{\gU}[1]{\opU\parens*{#1}}
\newcommand{\gSO}[1]{\opSO\parens*{#1}}

\DeclareMathOperator*{\volop}{vol}
\newcommand{\vol}[1]{\volop\parens{#1}}

\newcommand{\colorcomment}{Gray}
\algrenewcommand{\algorithmiccomment}[1]{\hfill\makebox[0.35\columnwidth][l]{{\color{\colorcomment}\scriptsize{$\blacktriangleright$ \textit{#1}}}}}

\algnewcommand\algorithmictake{\textbf{set}}
\algnewcommand\Take{\State \algorithmictake~}

\newcommand{\hadam}{\odot}

\newcommand{\agropt}[1]{#1}
\newcommand{\sepopt}{~}
\newcommand{\sepcon}{~}
\newcommand{\minp}[2]{\min_{#1}\sepopt{\agropt{#2}}}
\newcommand{\minpc}[3]{\min_{#1}\sepopt{\agropt{#2}}\sepcon\text{s.t. }#3}

\DeclarePairedDelimiter\abs{\lvert}{\rvert}
\newcommand{\absl}[1]{\abs*{#1}}

\newcommand{\subT}{_{\prn{T}}}

\newcommand{\btheta}{\boldsymbol\theta}

\DeclareMathOperator*{\argmin}{arg\,min}
\newcommand{\argminp}[2]{\argmin_{#1}\sepopt{\agropt{#2}}}

\newcommand{\subgr}{\mathcal{S}}
\newcommand{\graph}{\mathcal{G}}

\newcommand{\fpart}{f_{A,\bar{A}}}

\makeatletter
\def\ps@pprintTitle{%
 \let\@oddhead\@empty
 \let\@evenhead\@empty
 \def\@oddfoot{\centerline{\thepage}}%
 \let\@evenfoot\@oddfoot}
\makeatother

\newcommand{\plottorus}[2]{}
\newcommand{\noplottorus}{}

\newcommand{\plotovertorus}[1]{
\begin{tikzpicture}
    \begin{axis}[axis lines=none, view={45}{15}, axis equal, width=1.5\textwidth]
       \addplot3[scattergraph] table [row sep=newline, meta index=3] {#1};
       \addplot3[
        surf,
        samples=20,
        domain=0:2*pi,y domain=0:2*pi,
        z buffer=sort,
        line width=0.05pt,
        faceted color=torus,
        fill opacity=0.0,
        ]
       ({(2+cos(deg(x)))*cos(deg(y+pi/2))}, 
        {(2+cos(deg(x)))*sin(deg(y+pi/2))}, 
        {sin(deg(x))});
    \end{axis}
\end{tikzpicture}
}

\newif\iftikz
\tikzfalse

\newcommand{\tikzpath}{./Images/Tikz}
\newcommand{\pdfpath}{./}

\iftikz\tikzset{external/system call={/usr/local/texlive/2015/bin/x86_64-linux/lualatex -shell-escape -halt-on-error -interaction=batchmode -jobname "\image" "\texsource"}}\fi

\newcounter{Fig}

\newcommand{\includetikz}[1]{\iftikz\begin{minipage}[b]{\figwidth}\vspace{0pt}\scriptsize\centering\tikzincexp\tikzsetnextfilename{\tikznamee{#1}}\input{\tikzpath/#1.tikz}\end{minipage}\else\includetikzpdf{#1}\fi}
\newcommand{\includetikzex}[1]{\iftikz\includetikzextkz{#1}\else\includetikzexpdf{#1}\fi}

\newcommand{\tikznameexp}[1]{#1\tikzsuf}
\newcommand{\tikznameseq}[1]{\tikzbase\theFig\tikzsuf}

\newcommand{\includetikzpdf}[1]{\centering\tikzincexp\includegraphics{\pdfpath/\tikznamee{#1}.pdf}}
\newcommand{\includetikzextkz}[1]{\begin{minipage}[b]{\figwidth}\vspace{0pt}\scriptsize\centering\tikzincseq\tikzsetnextfilename{\tikznames{#1}}{#1}\end{minipage}}
\newcommand{\includetikzexpdf}[1]{\centering\tikzincseq\includegraphics{\pdfpath/\tikznames{#1}.pdf}}
\newcommand{\figwidth}{\columnwidth}
\newcommand{\tikzwidth}[1]{\renewcommand{\figwidth}{#1}}

\newcommand{\tikzincseq}{\stepcounter{Fig}}
\newcommand{\tikzincexp}{}

\newcommand{\tikzbase}{}
\newcommand{\tikzsuf}{}
\newcommand{\tikznamee}[1]{\tikznameexp{#1}}
\newcommand{\tikznames}[1]{\tikznameseq{#1}}

\newcommand{\tikzmodearxiv}{\renewcommand{\tikzbase}{Fig}\renewcommand{\tikzsuf}{_arXiv}\renewcommand{\tikznamee}[1]{\tikznameexp{##1}}\renewcommand{\tikznames}[1]{\tikznameseq{##1}}\setcounter{Fig}{0}}

\colorlet{graphic1}{BrickRed}
\colorlet{graphic2}{NavyBlue}
\colorlet{graphic3}{Goldenrod}
\colorlet{graphic4}{ForestGreen}
\colorlet{graphic5}{Peach}
\colorlet{torus}{Grey!80!White}

\newcommand{\showcolor}[1]{({\color{#1} $\bullet$})}

\begin{document}

\begin{frontmatter}

\title{Magnetic Eigenmaps for the Visualization of Directed Networks}
\author{Micha\"el Fanuel\corref{mycorrespondingauthor}}
\author{Carlos M. Ala\'iz}
\author{\'Angela Fern\'andez} 
\author{Johan A. K. Suykens}

\address{KU Leuven, Department of Electrical Engineering (ESAT), Kasteelpark Arenberg 10, B-3001 Leuven, Belgium}

\begin{abstract}
We propose a framework for the visualization of directed networks relying on the eigenfunctions of the magnetic Laplacian, called here Magnetic Eigenmaps.
The magnetic Laplacian is a complex deformation of the well-known combinatorial Laplacian. Features such as density of links and directionality patterns are revealed by plotting the phases of the first magnetic eigenvectors. An interpretation of the magnetic eigenvectors is given in connection with the angular synchronization problem.
Illustrations of our method are given for both artificial and real networks.
\end{abstract}


\end{frontmatter}

\tikzmodearxiv

\section{Introduction}
\label{SecIntro}

Many problems in neuroscience, biology, social or computer science are phrased in terms of networks and graphs.
The embedding of data points forming undirected graphs can be performed using manifold learning methods, among which are the so-called Laplacian Eigenmaps~\cite{BelkinEigenMaps} and Diffusion Maps~\cite{Coifman05geometricdiffusions}. In the same spirit, the embedding of a directed graph originating from the sampling of a vector field on a manifold  was studied in~\cite{PerraultDiscrete}. A Laplacian for strongly connected and aperiodic directed networks was introduced by Chung~\cite{Chung05laplacian} in relation with a random walk process, which was used for visualization e.g. in~\cite{Skillicorn}. Actually, Laplacians are very useful tools for community detection and data visualization. A common feature of these approaches is the relevance of the discrete or combinatorial Laplacian, and its normalized versions.

In this letter, no assumption on the origin of directed networks is needed, so we could deal, for example, with networks of webpages which are not embedded in any vector space.
In particular, we propose here the use of another Laplacian which naturally exists for a general connected directed network, called the magnetic Laplacian. This operator is actually a vector bundle Laplacian as described in~\cite{KenyonVectorBundle,FormanDeterminant} and a Connection Laplacian~\cite{SingerWu}. Interestingly, the magnetic Laplacian can be interpreted as a discrete quantum mechanical Hamiltonian of a charged particle on a network, influenced by a magnetic flux~\cite{Shubin,deVerdiere,Berkolaiko}. The method that we describe assigns a complex rotation, i.e., an element of $\gU{1}$, to each directed link, and the orientation of the link determines the direction of the rotation~\cite{CucuringuSyncRank}.

This letter is organized as follows. In Section~\ref{SecMagEig} the magnetic Laplacian and its eigenvectors are introduced. A method using the complex phase of these eigenvectors for visualizing directed graphs is proposed in Section~\ref{SecVisual}. Some examples are shown in Section~\ref{SecApp}, and the letter ends with some conclusions in Section~\ref{SecConcl}.

\section{Magnetic Laplacian and Eigenmaps}
\label{SecMagEig}

Consider a connected graph $\graph = (V,E)$ with a set of $N$ nodes $V$ and a set of undirected edges $E$.
In the case of an undirected graph, a symmetric weight matrix $\sym{W}$ is given with elements $\brq{\sym{W}}_{ij} = \sym{w}_{ij} \geq 0$ for all $i$ and $j\in V$. The Laplacian Eigenmaps are the eigenvectors of the combinatorial Laplacian $L\dz = D - \sym{W}$, where $D$ is the diagonal degree matrix with matrix elements $\brq{D}_{ii}  = d_i= \sum_{j \in V} \sym{w}_{ij}$ for all $i \in V$. The volume of a subgraph $\subgr_A$ of $\graph$ with node set $A$ is simply $\vol{\subgr_A} = \sum_{i \in A} d_i$.
In the case of directed networks, the graph is given by an asymmetric weight matrix $W$ with elements $\brq{W}_{ij} = w_{ij} \ge 0$. For simplicity, the  weights are chosen to be binary, i.e., $w_{ij} = 1$ if there is a link from $i$ to $j$ and $w_{ij} = 0$ otherwise. The weight matrix $W$ can be decomposed into a symmetric term  $\sym{w}_{ij} = \prn{w_{ij} + w_{ji}} / 2$, indicating that $\set{i, j}\in E$, and a skew-symmetric term, the edge flow $a_{ij} = - a_{ji}$ encoding the direction of the link. For all $\set{i,j}\in E$, we have $a_{ij} = 1$ if the link points from $i$ to $j$, and $a_{ij} = 0$ if $\set{i,j}$ is not directed.

In this letter, the magnetic Laplacian is defined as the self-adjoint, positive semi-definite operator $L\dg = D - T\dg \hadam \sym{W}$, where $D$ is the degree matrix associated with the symmetrized weight matrix, $0 \leq g < 1 / 2$ is an electric charge parameter, and $\brq{T\dg \hadam \sym{W}}_{ij} =\exp\prn{\im 2 \pi g a_{ji}} \sym{w}_{ij}$ (notice the Hadamard product $\odot$). The solutions of the generalized eigenvalue problem $L\dg \phi = \lambda D \phi$ are the \emph{Magnetic Eigenmaps} $\phi\dg_k$  associated with the eigenvalues $\lambda\dg_k \geq 0$ for $k \in \set{0, \dots, N-1}$~\cite{MagneticEigenMapsCommunityDetection} (we assume $\lambda\dg_0 \leq \lambda\dg_1 \leq \dots \leq \lambda\dg_{N-1}$).

\subsection{Interpretation of the first eigenvectors}

While the calculation of the \emph{second} eigenvector of the normalized combinatorial Laplacian is a relaxation of the (normalized) cut problem, the calculation of the \emph{first} eigenvector of the normalized magnetic Laplacian is a relaxation of the angular synchronization problem~\cite{singer2011angular}.
Given a subgraph $\subgr$ of $\graph$ (in general, one can choose $\subgr = \graph$), the angular synchronization problem consists in finding the angles $\btheta^\star = \prns{\theta^\star_1, \dotsc ,\theta^\star_N}^\intercal \in \gU{1}^{N}$ given by $\btheta^\star \in \argminp{\btheta}{\eta_{\subgr}\prns{\btheta}}$ where the frustration~\cite{bandeira2013cheeger} is defined by
\begin{equation*}
 \eta_{\subgr}\prns{\btheta} = \frac{1}{2} \frac{\sum_{i, j \in \subgr} \sym{w}_{ij} \absl{e^{\im \theta_i} - e^{\im \theta_{ij}} e^{\im \theta_j}}^2}{\sum_{i \in V} d_i} ,
\end{equation*}
with $\theta_{ij} = 2 \pi g a_{ji}$ for all $i, j \in \subgr$ such that $\sym{w}_{ij} \neq 0$. Notice that $\sum_{i \in V} d_i = \vol{\graph}$. The lowest eigenvector of the normalized magnetic Laplacian $\phi\dg_0$ is the solution of the spectral problem relaxing $\minp{\btheta}{\eta_{V}\prns{\btheta}}$. Our first conclusion is that computing the complex phase of $\phi\dg_0$ yields an approximation of $\btheta^\star$ that we propose to choose as the first visualization coordinate. In~\cite{CucuringuSyncRank}, the solution of the angular synchronization problem is shown to provide a ranking of the nodes in directed graphs, although a slightly different eigenvector problem is considered.

The performance of the spectral relaxation of the cut problem can be studied using a classical result of spectral graph theory, the Cheeger inequality, which relates the Cheeger constant to the \emph{second} smallest eigenvalue of the combinatorial Laplacian, providing the worst case performance for the spectral clustering method. Analogous results relate the \emph{first} smallest eigenvalue of the Connection Laplacian~\cite{bandeira2013cheeger} and the magnetic Laplacian~\cite{ShipingLiu} to a frustration quantifying the amount of inconsistency in the connection graph.
In particular, the performance naturally depends on the inverse of the spectral gap $1 / \lambda\dz_1$ of the undirected measurement graph. Indeed, the quality of the synchronization will benefit from a good connectivity of the nodes in the measurement graph. On the contrary, if $\lambda\dz_1$ is small, it could be instructive to find the subgraphs of $\graph$ where the frustration is minimal by cutting edges where the angular synchronization is not accurate.
Suppose that $A \subset V$ is the vertex set of a subgraph $\subgr_A$ of $\graph$ and let $\bar{A}$ be its complement in $V$.
A combinatorial graph partitioning problem is proposed, that is, $\minp{A}{E_{A, \bar{A}}\prns{\btheta^\star}}$ with
\begin{equation}
\label{eq:ProblemGeneralizedNCut}
 E_{A, \bar{A}}\prns{\btheta^\star} = \frac{\vol{\subgr_{\bar{A}}}}{\vol{\graph}} \eta_{A}\prns{\btheta^\star} + \frac{\vol{\subgr_A}}{\vol{\graph}} \eta_{\bar{A}}\prns{\btheta^\star} + \prn{\frac{c_{A, \bar{A}}}{\vol{\subgr_A}} + \frac{c_{\bar{A}, A}}{\vol{\subgr_{\bar{A}}}}} + \frac{\gamma_{A,\bar{A}}\prns{\btheta^\star}}{\vol{\graph}} ,
\end{equation}
where $\vol{\subgr_A} = \sum_{i \in A} d_i$ is the volume of $\subgr_A$, whereas the cut is $c_{A, \bar{A}} = c_{\bar{A}, A} = \sum_{i \in A, j \in \bar{A}} \sym{w}_{ij}$.
We have defined the generalized cut
\begin{equation}
\label{eq:GeneralizedCut}
 \gamma_{A,\bar{A}}\prns{\btheta^\star} = - 4 \sum_{i \in A, j \in \bar{A}} \sym{w}_{ij} \sin^2\prn{\frac{\theta^\star_i - \theta^\star_j - \theta_{ij}}{2}} .
\end{equation}
Notice that each term of~\eqref{eq:GeneralizedCut} is minimized when $\abs{\theta^\star_i - \theta^\star_j - \theta_{ij}} = \pi$, i.e., when the error made in the synchronization of rotations along the link $\set{i, j}$ is large.
Problem~\eqref{eq:ProblemGeneralizedNCut} generalizes the normalized cut problem aiming to find $A$ and $\bar{A}$ so that a combination of the frustration of both subgraphs, the normalized cut and the generalized cut~\eqref{eq:GeneralizedCut} is minimal, while the partition is balanced.
In order to construct a relaxation of this problem, we define
\begin{equation}
\label{eq:CutFunction}
 \prn{\fpart}_{i} = \begin{cases}
                    \hphantom{+} \sqrt{\vol{\subgr_{\bar{A}}}/\vol{\subgr_A}} e^{\im \theta^\star_i} & \text{ if } i \in A , \\
                    - \sqrt{\vol{\subgr_A}/\vol{\subgr_{\bar{A}}}} e^{\im \theta^\star_i} & \text{ if } i \in \bar{A} .
                   \end{cases}
\end{equation}
Then, we have $E_{A, \bar{A}}\prns{\btheta^\star} = \prns{\fpart^\dagger L\dg \fpart}/\prns{\fpart^\dagger D \fpart}$ as well as the relations $\sum_{i \in V} d_i e^{- \im \theta^\star_i} \prns{\fpart}_{i} = 0$ and $\sum_{i \in V} d_i \abs{\prns{\fpart}_{i}}^2 = \vol{\graph}$. As a consequence, a spectral relaxation of the combinatorial problem is given by
\begin{equation}
 \minpc{f\in \C^{N}_0}{\frac{f^\dagger L\dg f}{f^\dagger D f}}{f^\dagger D\phi\dg_0 = 0},
\end{equation}
which corresponds to the generalized eigenvalue problem for the second least eigenvector $\phi\dg_1$. In view of~\eqref{eq:CutFunction}, we conclude that $\phase{\phi\dg_0}$ will instruct us about the partition of the graph minimizing~\eqref{eq:ProblemGeneralizedNCut}. Indeed, we have $\phase{\prn{\fpart}_{i}} = \theta^\star_i$ if $i \in A$ and $\phase{\prn{\fpart}_{i}} = \theta^\star_i + \pi$ if $i \in \bar{A}$. The eigenvector $\phi\dg_{1}$ being the relaxed version of $\prn{\fpart}_{i}$, we expect that, by embedding the graph with $\phase{\phi\dg_0}$ as first coordinate and $\phase{\phi\dg_1}$ as second coordinate, we obtain two parallel groups distant by $\pi$.
By analogy with the Laplacian eigenmaps, similar results are expected for the next eigenvalues.

\subsection{Relation with the combinatorial Laplacian}

When the edge flow of the network is given exactly by a certain potential $h$, i.e., $a_{ij} = h_j - h_i$ for all $\set{i ,j} \in E$, the spectrum of the magnetic Laplacian corresponds to the spectrum of the combinatorial Laplacian, and the eigenvectors are related by $\phi\dg_{k,i} = \exp\prns{\im 2 \pi g h_i} \phi\dz_{k,i}$ for all $i \in V$ and all $k \in\set{0,\dots, N-1}$. This particular case is characterized by the first eigenvalue, as stated in the following proposition, which can be seen as a consequence of~\cite[Theorem~2.6]{bandeira2013cheeger} and of~\cite{Berkolaiko,deVerdiere} in the context of mathematical physics.
\begin{myprp}
\label{ThZeroEig}
Consider a connected graph $\graph$. The magnetic Laplacian $L\dg$ has a zero eigenvalue iff there exists a function $h$ satisfying, for any link $\set{i,j}\in E$, $a_{ij} = h_j - h_i$.
\end{myprp}
Nonetheless, since we assume here that $a_{ij} \in \set{-1, 0, 1}$ for all $\set{i,j} \in E$, this situation only happens in our context if the graph is a tree.
The relation in the general case can be characterized by Lemma~\ref{TreeLemma}.
First let us define the parallel transporter $t\dg_P \in \gU{1}$ over a finite path $P = \set{i_1, i_2, \dotsc, i_n}$ of length $n > 1$ in the graph $\graph$ as $t\dg_P = \exp\prns{\im 2 \pi g \prns{a_{i_1 i_2} + \dots + a_{i_{n-1} i_n}}}$.
Moreover, following~\cite{ChungKempton} we call a directed graph $\epsilon$-consistent if, for every simple cycle $\set{i_1, i_2, \dots, i_n, i_{n+1} = i_1}$, we have $\abs{t\dg_C - 1} \leq \epsilon$, with $t\dg_C = \exp\prn{\im 2 \pi g \oint_C a}$, where the magnetic flux is defined by the discrete line integral $\oint_C a \triangleq a_{i_1 i_2} + \dots + a_{i_n i_1}$. We now state an elementary result given in~\cite{Berkolaiko}.
\begin{mylma}
\label{TreeLemma}
 Consider a connected directed graph $\graph$, and let $T$ be any spanning tree of $\graph$, and $\bar{T} = E \setminus T$. For all $0 \leq g \leq 1 / 2$, the magnetic Laplacian $L\dg$ is unitarily equivalent to the operator given by
 \begin{equation*}
 \label{eq:UnitaryTransf}
  \prn{\tilde{L}\dg\subT f}_i = \sum_{\set{j | \set{i,j} \in T}} \sym{w}_{ij} \prn{f_i - f_j} + \sum_{\set{j | \set{i,j} \in \bar{T}}} \sym{w}_{ij} \prn{f_i - t\dg_{\circ ,ij} f_j} ,
 \end{equation*}
 for all $f_i \in \C$ and all $i \in V$. The operator $t\dg_{\circ, ij}$ is defined as $t\dg_{\circ, ij} = t\dg_{i_0 \to i} t\dg_{ij} t\dg_{j \to i_0}$ for any $i_0 \in V$, and $i_0 \to k$ denotes the unique path in the tree $T$ going from $i_0 \in T$ to $k \in T$. The unitary transformation, realizing $\tilde{L}\dg\subT = U\subT^\dagger L\dg U\subT$, is given by the diagonal matrix $U\subT$, with elements $\brq{U\subT}_{ii} = t\dg_{i_0 \to i}$ for all $i \in T$.
\end{mylma}
Notice that the accumulated complex phase $t\dg_{\circ, ij}$ (holonomy or magnetic flux) in the contour integral does not depend on the choice of $i_0$. Its value is equal to the contour integral on the smallest loop including the link $\set{i, j}$.
Moreover, as a corollary, we have that the normalized magnetic Laplacian $L\dg\nr = D^{-1/2} L\dg D^{-1/2}$ is unitarily equivalent to $D^{-1/2} \tilde{L}\dg\subT D^{-1/2}$.
Hence, if the graph $\graph$ is a tree, the spectrum of the magnetic Laplacian is exactly the spectrum of the combinatorial Laplacian and their eigenvectors are in one to one correspondence, thanks to diagonal unitary transformation.

\section{Visualization of density and directionality}
\label{SecVisual}

The eigenmaps can be calculated by computing the eigenvectors of the normalized magnetic Laplacian $L\dg\nr = D^{-1/2} L\dg D^{-1/2}$ and, therefore, these eigenvectors are obtained by calculating the largest eigenvalues of $D^{-1/2} T\dg \hadam \sym{W} D^{-1/2}$.
Once the Laplacian is normalized, we propose to embed the network by the mapping $i \mapsto \prns{\phase{\phi_{0,i}\dg}, \phase{\phi_{1,i}\dg}, \dotsc, \phase{\phi_{n,i}\dg}}^\intercal$.
Notice that the phase operator identifies the angles that differ by $2 \pi$, which means that the geometrical representation, for the $1$-dimensional case, is just a circle, whereas for the general $n$-dimensional case is an $n$-torus.
Hence, for visualization purposes, the directed network will be embedded on a $2$-torus represented as the square $[0, 2\pi] \times [0, 2\pi]$ with opposite sides identified, as shown in Figure~\ref{FigSchemes}.
Therefore, the visualization will be symmetric if an eigenvector undergoes a global rotation in the complex plane. Notice that given an eigenvector $\phi_k\dg$ of the magnetic Laplacian, another eigenvector of the same eigenvalue is given by $e^{\im \alpha} \phi\dg_k$.
We will show empirically that this low dimensional embedding is able to visualize at the same time dense regions of links, revealed by the $y$-axis, and patterns determined by the link directions, given by the $x$-axis.

\begin{figure}
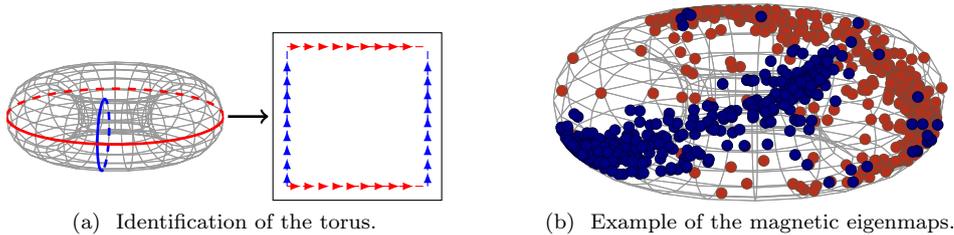

 \centering
 \subfloat[\label{FigUnrolledTorus} Identification of the torus.]{\tikzwidth{0.45\textwidth}\plottorus{}{}\includetikz{TorusUnrolled}}\quad%
 \subfloat[\label{FigPolBlogsTorus} Example of the magnetic eigenmaps.]{\begin{minipage}[b]{0.45\textwidth}\tikzwidth{0.8\textwidth}\centering\includetikzex{\plotovertorus{./Images/Data/Torus_polBlogs_14.txt}}\end{minipage}}
 \caption{\label{FigSchemes} \protect\subref{FigUnrolledTorus}: Representation of the $2$-torus as the square $[0, 2\pi] \times [0, 2\pi]$ with identified sides; the position of the cuts is arbitrary and can be adapted to each particular dataset.
 \protect\subref{FigPolBlogsTorus}: Example of the magnetic eigenmaps plotted over the $3$-dimensional $2$-torus for a graph with two communities, corresponding to the dataset ``Political Blogosphere'' explained in detail in Section~\ref{SecApp}.}
\end{figure}

\subsection{Consistency of the mapping}

The visualization method relies only on the phases of the magnetic eigenmaps, which contain most of the information if the magnetic Laplacian is unitarily equivalent to the combinatorial Laplacian.
We will characterize now the error made if this is not the case. A bound on the variance of $\abs{\phi\dg_0}$ is obtained from \cite[Lemma~3.3]{bandeira2013cheeger} that is adapted below to the setting of this paper. For all $z \in \C^N$, we define $\tilde{z} \in \C^N$ such that $\tilde{z}_i = z_i / \abs{z_i}$ if $z_i \neq 0$ and $\tilde{z}_i = 0$ if $z_i = 0$.
\begin{mylma}[Bandeira \emph{et al.}~{\cite[Lemma~3.3]{bandeira2013cheeger}}]
 For all $z\in \C^N$, we have $\sum_i d_i \abs{z_i - \mu \tilde{z}_i}^2 \leq \eta\prn{z} \sum_i d_i \abs{z_i}^2 / \lambda_1\dz$, with $\mu = \prns{1/\vol{\graph}} \sum_j d_j \abs{z_j}$ and $\eta\prns{z} = z^\dagger L\dg z / \prns{z^\dagger D z}$.
\end{mylma}
Hence, by taking $z = \phi\dg_0$, we have a bound on the variability of $\abs{\phi\dg_0}$. More precisely, the variation $\abs{\phi\dg_0}$ with respect to its mean value is constrained as follows:
\begin{equation}
\label{eq:BoundLambda0}
 \frac{\sum_i d_i \absl{\abs{\phi\dg_{0,i}} - \mu_0}^2}{\sum_i d_i \abs{\phi\dg_{0,i}}^2} \leq \frac{\lambda\dg_0}{\lambda\dz_1}, \text{ with } \mu_0 = \frac{1}{\vol{\graph}} \sum_j d_j \abs{\phi\dg_{0,j}} .
\end{equation}
Therefore, since the spectral gap $\lambda\dz_1$ is only determined by the density of the undirected graph, the previous bound on the variability of the modulus $\abs{\phi\dg_0}$ can only be reduced if the eigenvalue $\lambda\dg_0$ is made smaller.
In particular, Theorem~1 in~\cite{ChungKempton} can be adapted to the case of the magnetic Laplacian, so that the smallest eigenvalue of the magnetic Laplacian of an $\epsilon$-consistent graph satisfies $\lambda\dg_0 \leq \epsilon^2 / 2$.
In our particular case, this bound can be improved by relating the inconsistency to a topological property of the graph leading to possible inconsistencies: the first Betti number $\beta_1 = \abs{E} - \abs{V} + 1$, i.e., the number of simple cycles in the graph. Notice that there are exactly $\beta_1$ edges to remove to the graph in order to have a tree.
\begin{mylma}
\label{lem:IneqLambda0}
 Consider a connected directed graph $\graph$, and let $T$ be any spanning tree of $\graph$, and $\bar{T} = E \setminus T$. We have the following inequality
 \begin{equation}
 \label{eq:IneqLambda0}
  \lambda\dg_0 \leq \epsilon_{g}^2 \frac{\sum_{\set{i, j} \in \bar{T}} \sym{w}_{ij}}{\vol{\graph}} ,
 \end{equation}
  where the summation includes only $\beta_1$ terms (associated with the number of simple cycles).
  In particular, if the weights are binary $ \lambda\dg_0 \leq \frac{\epsilon_{g}^2 \beta_1}{2 \abs{E}}$.
\end{mylma}
\begin{proof}
 Fix a spanning tree $T$ of $\graph$. Using Lemma~\ref{TreeLemma}, we can consider the smallest eigenvalue of the unitary equivalent operator $\tilde{L}\dg\subT$, which is the minimum of the Rayleigh quotient $R\prns{f} = \sum_{i} f_i^\dagger\prns{\tilde{L}\dg\subT f}_i / \sum_{i} d_i \abs{f_i}^2$. Specifically, choosing $f_i = 1$ for all $i \in V$, we have
 \begin{equation*}
  \lambda\dg_0 \leq R(f) = \frac{\sum_{\set{i, j} \in \bar{T}} \sym{w}_{ij} \abs{f_i - t\dg_{\circ,ij} f_j}^2}{\sum_{i \in V} d_i} = \frac{\sum_{\set{i, j} \in \bar{T}} \sym{w}_{ij} \abs{1 - t\dg_{\circ, ij}}^2}{\vol{\graph}} \le \epsilon_{g}^2 \frac{\sum_{\set{i, j} \in \bar{T}} \sym{w}_{ij}}{\vol{\graph}} .
 \end{equation*}
\end{proof}

The bound given in \eqref{eq:BoundLambda0} suggests to reduce the variability of the modulus (the information lost when considering only the phases) by making $\lambda\dg_0$ smaller. According to Lemma~\ref{lem:IneqLambda0}, this can be done looking for the $g$ that minimizes the inconsistency, which is a task dependent on each particular graph.
Moreover, the results above are only upper bounds, so there is no guarantee that such a value of $g$ is the best possible choice.
Finally, notice that we are interested in maximizing the information included in the phase of the eigenvector, not only in minimizing the information lost while ignoring the modulus. As a trivial counterexample, if we take $g \to 0$ then we will recover the combinatorial Laplacian, and hence $\lambda\dg_0 \to 0$. Although in that case the eigenvector will be constant in modulus, and no information will be lost, its phase will also be constant, and thus it will provide no information.

\subsection{Selection of \texorpdfstring{$g$}{g}}
\label{SecSG}

The choice of the electric charge parameter $g$ influences the visualization method. As stated above, there is not an established method to select it.
In general, we propose to choose a quantized charge $g = k/m$ with $k \notin m \Z$. The particular value $g = 1/3$ is suited in the presence of directed triangles, while $g = 1/4$ is relevant in the presence of directed $4$-cycles. For $g = 2 / 5$, the magnetic Laplacian becomes quite similar with the signed Laplacian, while for $g = 1 / 2$ it is a signed Laplacian associated with the same graph where undirected edges are labelled as $(+)$ and directed edges as $(-)$. For more details we refer to~\cite{MagneticEigenMapsCommunityDetection}. Notice that choosing $g > 1 / 2$ would be equivalent to a flip of all link directions.

\subsection{Connection with Vector Diffusion Maps}
\label{SecVDM}

The computation of the eigenvectors of the normalized magnetic Laplacian can resemble the Vector Diffusion Maps of Singer and Wu~\cite{SingerWu}, although in our case we work with a complex and unitary transporter in $\gU{1}$, instead of with an orthogonal transporter. 
In particular, the main difference between both approaches, apart from the working space ($\R^n$ and $\C$, respectively), resides in the transport term. In the case of Vector Diffusion Maps, it is an element of $\gSO{n}$ determined by Local PCA, and for Magnetic Eigenmaps it is an element of $\gU{1}$ which can be tuned by the user through the parameter $g$ in order to highlight certain properties of the graph.
Moreover, the methodology of both approaches differs in the way they map the data.
On the one hand, Vector Diffusion Maps follows a natural extension of Diffusion Maps~\cite{Coifman05geometricdiffusions}, so that each point is mapped to a matrix defined in terms of the eigenvalues and eigenvectors of the transition matrix.
On the other hand, the proposed magnetic eigenmaps map the points to the phases of the first eigenvectors of the corresponding Laplacian.
Therefore, although both methods share some similarities, they are essentially different, and none of them can be seen as a particular case of each other.

\section{Applications}
\label{SecApp}

We will illustrate now how Magnetic Eigenmaps can be applied to the visualization of directed graphs on two synthetic and two real networks.
For completeness, we include the procedure to compute the magnetic eigenmaps for a directed graph with binary weights in Algorithm~\ref{AlgMethod}.
In all the examples we will depict the aspect of the original network as a baseline, using for this purpose expert knowledge or a force-directed layout (based on using attractive forces between adjacent nodes and repulsive forces between distant nodes).
We will also compare with Diffusion Maps for the real examples. In this context, the embedding is obtained by computing the algorithm with $g = 0$, and plotting the first eigenvectors of the corresponding Laplacian (instead of their phases).
For the real networks we will also depict the spectrum.

\begin{algorithm}[t]
\caption{\label{AlgMethod} Magnetic Eigenmaps Visualization}
\begin{footnotesize}
\begin{algorithmic}[0]
 \Procedure{MEigenmaps}{$W, g$}
  \State $\sym{W} \gets \prn{W + W^\intercal} / 2$ \Comment{Symmetric weights.}
  \State $A \gets W - W^\intercal$ \Comment{Edge flow.}
  \State $d_{ii} \gets \sum_j {\sym{w}_{ij}}$ \Comment{Degree matrix.}
  \State $t\dg_{ij} \gets e^{\im 2 \pi g a_{ji}}$ \Comment{Transporter.}
  \State $L\dg \gets D - \sym{W} \hadam T\dg$ \Comment{Magnetic Laplacian.}
  \State $L\dg\nr \gets D^{-1/2} L\dg D^{-1/2}$ \Comment{Normalized Laplacian.}
  \State $\phi_0\dg, \phi_1\dg, \dotsc \gets \textproc{Eigs}\prns{L\dg\nr}$ \Comment{Eigenvectors.}
  \State\Return $\set{\phase{\phi_m\dg}}_{m = 0}^n$ \Comment{Phases.}
 \EndProcedure
\end{algorithmic}
\end{footnotesize}
\end{algorithm}

\subsection{Artificial networks}

We propose to visualize first the artificial network with a running flow of Figure~\ref{FigAffGroupsOrig}, where the coordinates of the nodes in the real plane have been chosen according to our knowledge about the underlying groups. This network, constructed according to~\cite{Lancichinetti}, consists of three groups of ten nodes ($A$, $B$ and $C$). Two nodes in the same group are linked with a probability $0.5$. Any node has also a probability $0.5$ to be connected to a node from another group. Furthermore $90$ percent of these interconnections are directed in the direction of the flow, i.e., $A \to B$, $B \to C$ and $C \to A$.
Plotting the real and imaginary parts of the first eigenvector of the magnetic Laplacian can indicate the presence of a running flow in the network, as illustrated in Figure~\ref{FigAffGroupsReIm}.
However there could also be dense clusters in the network that are not revealed using only the phase of the first eigenfunction.
In order to actually visualize the three groups and the density information, we use our proposed method of plotting the complex phase of the first eigenfunction versus the phase of the second eigenfunction of the magnetic Laplacian, as illustrated in Figure~\ref{FigAffGroupsME}. Notice that the phase of the second eigenfunction does not distinguish specific dense clusters in the network, while the phase of the first eigenfunction (corresponding to directionality) is able to separate the three groups.

\begin{figure}[t]
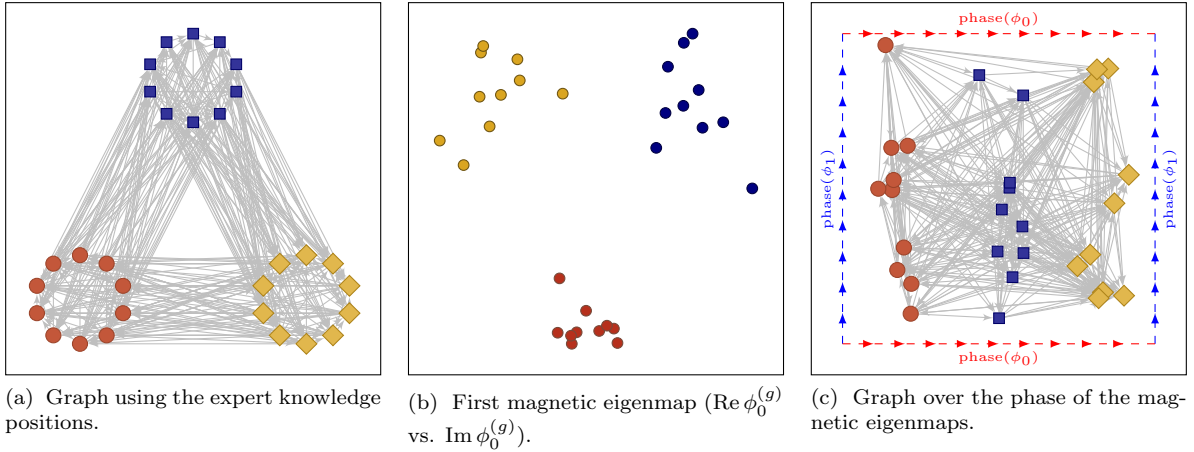

 \centering
 \tikzwidth{0.3\textwidth}
 \subfloat[\label{FigAffGroupsOrig} Graph using the expert knowledge positions.]{\noplottorus\includetikz{Graph_AffGroups_EK}}\quad%
 \subfloat[\label{FigAffGroupsReIm} First magnetic eigenmap ($\real{\phi\dg_0}$ vs. $\imag{\phi\dg_0}$).]{\includetikz{ReIm_AffGroups}}\quad%
 \subfloat[\label{FigAffGroupsME} Graph over the phase of the magnetic eigenmaps.]{\plottorus{0}{1}\includetikz{Graph_AffGroups_ME}}
 \caption{\label{FigAffGroups} Artificial network with a running flow. The colours indicate the three groups, and the magnetic eigenmaps correspond to $g = 1/4$.}
\end{figure}

We are going to consider now an example of a network with a small number of nodes playing a particular role and then a clear structure with two dense clusters. In~\cite{LeichtNewman}, an artificial network of $32$ nodes is built as follows: it consists of two dense groups of $14$ nodes with a few interconnecting links and two pairs of nodes are connected to the whole network. The first pair has only in-coming links, while the second pair has only out-going links.
An illustration using the force layout is given in Figure~\ref{FigNewmanOrig}. Plotting the real and imaginary parts of the first eigenfunction allows to distinguish the two pairs from the rest of the network, as showed in Figure~\ref{FigNewmanReIm}.
However, it is more instructive to visualize the network using the phase of the two first eigenfunctions of the magnetic Laplacian. Indeed, the two groups and the two pairs are easily separated in Figure~\ref{FigNewmanME}, where the phase of the first eigenfunction (directionality) is able to separate the two pairs of disconnected points from the rest of the set, defining three directionality-groups: green points, the yellow points and the blue and red points together. On the other side, the phase of the second eigenfunction (corresponding to density information) shows two groups: the blue and green points versus the red and yellow points. Combining the information given by the two phases we are able to easily separate visually the four groups that we were looking for.

\begin{figure}[t]
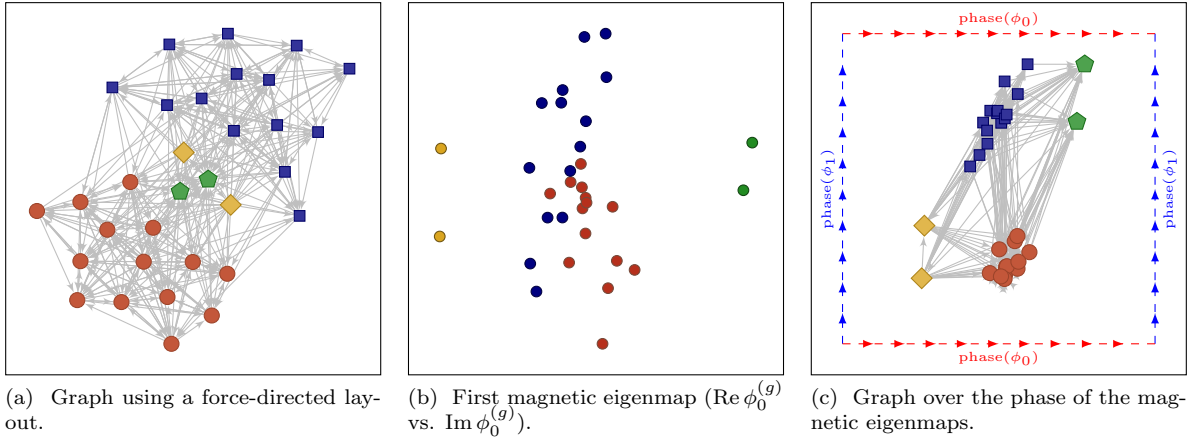

 \centering
 \tikzwidth{0.3\textwidth}
 \subfloat[\label{FigNewmanOrig} Graph using a force-directed layout.]{\noplottorus\includetikz{Graph_Newman_Auto}}\quad%
 \subfloat[\label{FigNewmanReIm} First magnetic eigenmap (\smash{$\real{\phi\dg_0}$} vs. \smash{$\imag{\phi\dg_0}$}).]{\includetikz{ReIm_Newman}}\quad%
 \subfloat[\label{FigNewmanME} Graph over the phase of the magnetic eigenmaps.]{\plottorus{0}{1}\includetikz{Graph_Newman_ME}}
 \caption{\label{FigNewman} Artificial network with two dense clusters and two pairs of nodes with a specific role. The colours indicate the dense clusters and the hub pairs, and the magnetic eigenmaps correspond to $g = 1/4$.}
\end{figure}

\subsection{Directed networks from real data}

In the previous section, we considered directed networks having known structures either in terms of link directions or link density. Indeed, Magnetic Eigenmaps is able to provide simultaneously information about direction and density, as we illustrate now also on real directed networks where these two aspects are relevant.
Let us emphasize that we do not pretend that Magnetic Eigenmaps will give the best possible result on each dataset. Knowing the nature of the datasets, an \emph{ad hoc} visualization method will certainly give a better result. We will however assume that the origin of the dataset is unknown in order to show that Magnetic Eigenmaps indeed reveals relevant features.

The network used represents the \emph{common adjective and noun adjacencies} for the novel ``David Copperfield'' by Charles Dickens~\cite{NewmanPREEigenvectors}. This directed graph has $112$ nodes that represent the most commonly occurring adjectives and nouns in the book. Edges connect any pair of words that occur in adjacent position in the text of the book. From the structure of English, a certain directional structure can be anticipated, i.e., adjectives are expected to be found before nouns.
The graph representation of this dataset, using a force layout, is shown in Figure~\ref{FigWordAdjOrig}, where the structure can hardly be guessed. However, considering the phase of the two first magnetic eigenmaps (we have used $g = 2 / 5$ since the graph is almost bipartite), it is possible to visualize the presence of two groups, as illustrated in Figure~\ref{FigWordAdjME} (indeed, the information is provided mostly by the first coordinate, corresponding to directionality).
Finally, we show in Figure~\ref{FigWordAdjDM} the first embedding coordinates in the Diffusion Maps case (using the normalized Laplacian), more concretely the second and third eigenvectors (the first one is discarded because it is constant). In this case both classes appeared mixed, so these two diffusion coordinates are not able to reveal the structure of the data. The reason is that Diffusion Maps captures the density of the links but not the directionality, while Magnetic Eigenmaps relates both.
Additionally, Figure~\ref{FigWordAdjEigs} represents the first eigenvalues of the combinatorial ($g = 0$) and magnetic Laplacian ($g = 2/5$). The increase of these eigenvalues is smooth, presenting just an eigengap between the first eigenvalue and the second one. This remark motivates the choice of  the first and second eigenvectors for this example as visualization coordinates. Moreover, the differences between the two spectra, and the nonzero initial eigenvalue, show that the structure of the graph is not trivial (see Proposition~\ref{ThZeroEig}).

\begin{figure}[t]
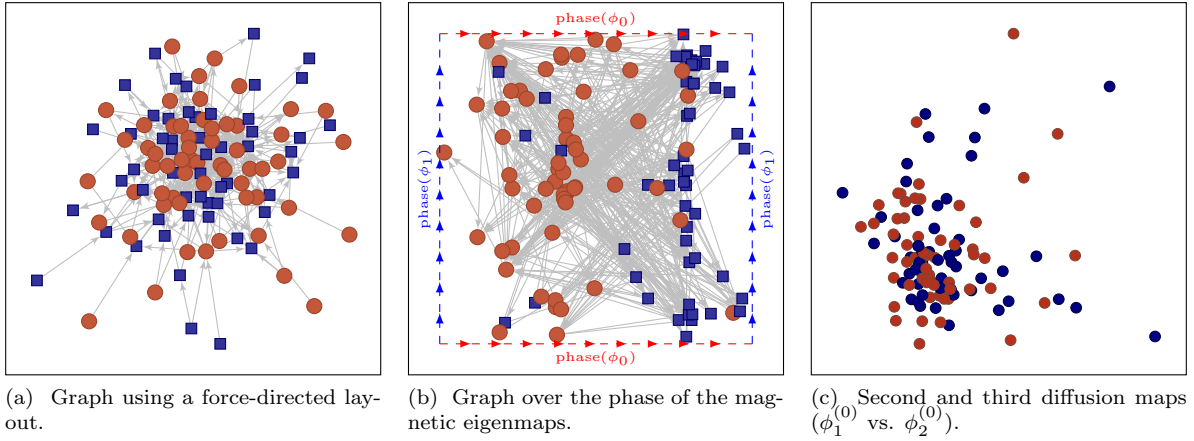

 \centering
 \tikzwidth{0.3\textwidth}
 \subfloat[\label{FigWordAdjOrig} Graph using a force-directed layout.]{\noplottorus\includetikz{Graph_WordAdjacencies_Force}}\quad%
 \subfloat[\label{FigWordAdjME} Graph over the phase of the magnetic eigenmaps.]{\plottorus{0}{1}\includetikz{Graph_WordAdjacencies_ME}}\quad%
 \subfloat[\label{FigWordAdjDM} Second and third diffusion maps (\smash{$\phi\dz_1$} vs. \smash{$\phi\dz_2$}).]{\includetikz{DM_WordAdjacencies}}
 \caption{\label{FigWordAdj} Word adjacencies example. The colours indicate the class labels, nouns~\showcolor{graphic1} and adjectives \showcolor{graphic2}, and the magnetic eigenmaps correspond to $g = 2/5$.}
\end{figure}

\begin{figure}[t]
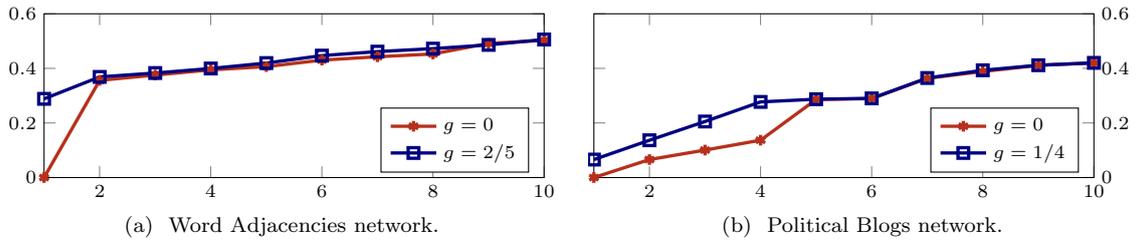

 \centering
 \tikzwidth{0.4\textwidth}
 \subfloat[\label{FigWordAdjEigs} Word Adjacencies network.]{\includetikz{Eigs_WordAdjacencies}}\quad%
 \subfloat[\label{FigPolBlogsEigs} Political Blogs network.]{\includetikz{Eigs_polBlogs}}
 \caption{First eigenvalues of the combinatorial and normalized magnetic Laplacian for the real networks.}
\end{figure}

The second real dataset used in these experiments represents the \emph{political blogosphere} in February of 2005~\cite{polBlogsDataSet}. This directed graph is composed of \num{1222} nodes that indicate the political leaning, meaning left or liberal and right or conservative (disconnected points were removed). The data on political leaning comes from blog directories and some of the blogs were labelled manually, based on incoming and outgoing links and posts around the time of the 2004 presidential election in the USA. The links between blogs were automatically extracted from a crawl of the front page of the blog.
>From Figure~\ref{FigPolBlogsOrig}, where the network has been depicted using the force layout, it is already possible to guess the presence of two dense groups of webpages. The first eigenvalues of the magnetic Laplacian in Figure~\ref{FigPolBlogsEigs} instruct us to consider the two first pairs of eigenvalues in order to visualize two different structures.
In Figures~\ref{FigPolBlogsME12} and~\ref{FigPolBlogsME13}, the magnetic eigenmaps do not distinguish the two classes of nodes, however we observe that some webpages are less connected to the rest of the network, whereas in Figure~\ref{FigPolBlogsME14} we see the two classes clearly separated. The latter mapping is also shown in $\R^3$ over the torus as an illustration in Figure~\ref{FigPolBlogsTorus}.
For comparison purposes, we show also in this case the first embedding diffusion coordinates in Figures~\ref{FigPolBlogsDM23} and~\ref{FigPolBlogsDM45}. In this example, where the graph structure is clearer than in the previous dataset, Diffusion Maps is able to condense the two classes separately just using the density of the graph. We would like to highlight that the distinction between both classes is not very clear when we select the second and third eigenvectors (the first eigenvector is again discarded), whereas a cleaner classification structure is obtained when the fourth and fifth eigenvectors are depicted. Nevertheless, Magnetic Eigenmaps represents well both the connectivity and density structure.

\begin{figure}[t]
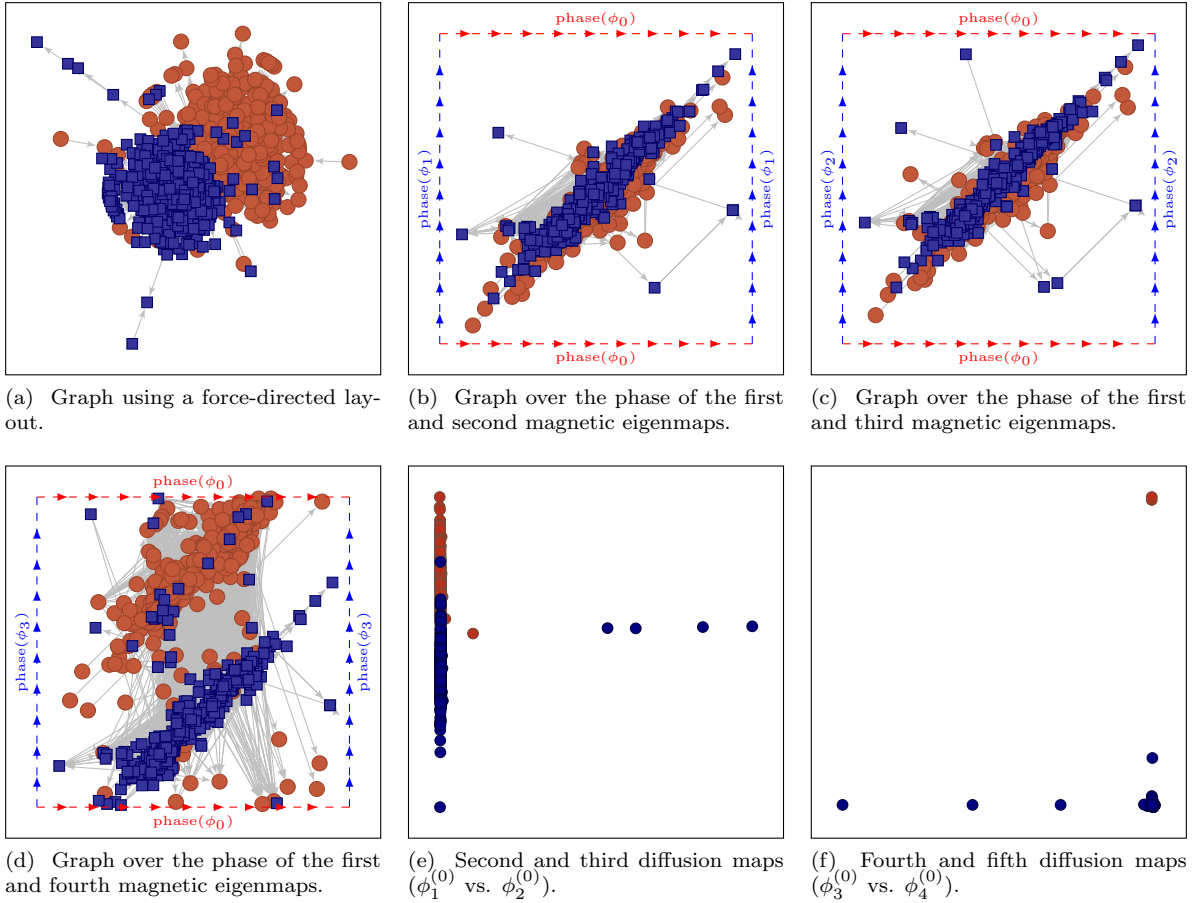

 \centering
 \tikzwidth{0.3\textwidth}
 \subfloat[\label{FigPolBlogsOrig} Graph using a force-directed layout.]{\noplottorus\includetikz{Graph_polBlogs_Force}}\quad%
 \subfloat[\label{FigPolBlogsME12} Graph over the phase of the first and second magnetic eigenmaps.]{\plottorus{0}{1}\includetikz{Graph_polBlogs_ME12}}\quad%
 \subfloat[\label{FigPolBlogsME13} Graph over the phase of the first and third magnetic eigenmaps.]{\plottorus{0}{2}\includetikz{Graph_polBlogs_ME13}}\\
 \subfloat[\label{FigPolBlogsME14} Graph over the phase of the first and fourth magnetic eigenmaps.]{\plottorus{0}{3}\includetikz{Graph_polBlogs_ME14}}\quad%
 \subfloat[\label{FigPolBlogsDM23} Second and third diffusion maps ($\smash{\phi\dz_1}$ vs. $\smash{\phi\dz_2}$).]{\includetikz{DM_polBlogs_23}}\quad%
 \subfloat[\label{FigPolBlogsDM45} Fourth and fifth diffusion maps ($\smash{\phi\dz_3}$ vs. $\smash{\phi\dz_4}$).]{\includetikz{DM_polBlogs_45}}\\
 \caption{\label{FigPolBlogs} Political blogosphere example. The colours indicate the class labels, left leaning \showcolor{graphic1} and right leaning \showcolor{graphic2}, and the magnetic eigenmaps correspond to $g = 1/4$.}
\end{figure}

\section{Conclusions}
\label{SecConcl}

In this letter, we have proposed the use of the eigenvectors of the magnetic Laplacian, called here \emph{Magnetic Eigenmaps}, for the visualization of directed networks. Our work is a natural extension of the Laplacian Eigenmaps.
Computationally, the method reduces to the calculation of the eigenvectors of maximal eigenvalues of a Hermitian matrix, which can be conveniently performed thanks to e.g. the power method.
The advantages of this approach were illustrated on artificial and real datasets, showing that our method is able to reveal both the directionality and connectivity patterns of the networks.

\section*{Acknowledgements}

\begin{small}
The authors thank the following organizations.
\begin{itemize*}
 \item EU: The research leading to these results has received funding from the European Research Council under the European Union's Seventh Framework Programme (FP7/2007-2013) / ERC AdG A-DATADRIVE-B (290923). This paper reflects only the authors' views, the Union is not liable for any use that may be made of the contained information.
 \item Research Council KUL: GOA/10/09 MaNet, CoE PFV/10/002 (OPTEC), BIL12/11T; PhD/Postdoc grants.
 \item Flemish Government:
 \begin{itemize*}
  \item FWO: G.0377.12 (Structured systems), G.088114N (Tensor based data similarity); PhD/Postdoc grants.
  \item IWT: SBO POM (100031); PhD/Postdoc grants.
 \end{itemize*}
 \item iMinds Medical Information Technologies SBO 2014.
 \item Belgian Federal Science Policy Office: IUAP P7/19 (DYSCO, Dynamical systems, control and optimization, 2012-2017).
\end{itemize*}
\end{small}


\bibliographystyle{elsarticle-num}
\bibliography{References}

\end{document}